\documentclass{amsart}
        \usepackage{amsmath}
        \usepackage{amsfonts}
        \usepackage{amsthm}
        \usepackage{comment}
        \usepackage{url}
        \usepackage{color}
        \usepackage{enumitem}
        \setlist[enumerate]{font=\normalfont}
        \usepackage{mathtools}
		\usepackage[shortalphabetic]{amsrefs}
        \usepackage{latexsym}
        \usepackage{amssymb}
		\usepackage{graphicx}        
		\usepackage{float}
        
        \usepackage[colorinlistoftodos]{todonotes}
		\usepackage[colorlinks=true, allcolors=blue]{hyperref}
		\usepackage{subcaption}

        \usepackage{eucal}

        \theoremstyle{plain}
        \newtheorem*{theorem*}{Theorem}
        \newtheorem{theorem}{Theorem}[section]
        \newtheorem{corollary}[theorem]{Corollary}
        \newtheorem{lemma}[theorem]{Lemma}

        \newtheorem{proposition}[theorem]{Proposition}

        \theoremstyle{definition}
        \newtheorem{definition}[theorem]{Definition}
        \newtheorem{example}[theorem]{Example}
        
        \newtheorem*{example*}{Example}

        \theoremstyle{remark}
        \newtheorem{remark}[theorem]{Remark}
        \newtheorem*{remark*}{Remark}

\newcommand{\PP}{\mathbb{P}}

\newcommand{\F}{\mathbb{F}}

\newcommand{\FF}{{\mathbb F}}

\DeclareMathOperator{\calC}{\mathcal{C}}

\author[T. Bogart]{Tristram Bogart}
\address{
Tristram Bogart\\
Departamento de matem\'aticas\\
Universidad de los Andes\\
Carrera $1^{\rm ra}\#18A-12$\\ 
Bogot\'a, Colombia
}
\email{tc.bogart22@uniandes.edu.co}

\author[A.-L. Horlemann]{Anna-Lena Horlemann-Trautmann}
\address{Anna-Lena Horlemann-Trautmann, Faculty of Mathematics and Statistics, University of St.\ Gallen, Switzerland}
\email{anna-lena.horlemann@unisg.ch}

\author[D. Karpuk]{David Karpuk}
\address{David Karpuk, Artificial Intelligence Centre of Excellence, F-Secure Corporation, Helsinki, Finland}  
\email{david.karpuk@f-secure.com}

\author[A. Neri]{Alessandro Neri}
\address{Alessandro Neri, Institute for Communications Engineering,
Technical University of Munich, Germany}
\email{alessandro.neri@tum.de}

\author[M. Velasco]{Mauricio Velasco}
\address{
Mauricio Velasco\\
Departamento de matem\'aticas\\
Universidad de los Andes\\
Carrera $1^{\rm ra}\#18A-12$\\ 
Bogot\'a, Colombia
}
\email{mvelasco@uniandes.edu.co}


\title[PMDS codes from reducible algebraic curves]{Constructing Partial MDS codes from reducible algebraic curves}

\begin{document}
\maketitle

\begin{abstract} We propose reducible algebraic curves as a mechanism to construct Partial MDS (PMDS) codes geometrically. We obtain new general existence results, new explicit constructions and improved estimates on the smallest field sizes over which such codes can exist. Our results are obtained by combining ideas from projective algebraic geometry, combinatorics and probability theory.\end{abstract}

\section{Introduction} \label{sec:intro}
Currently, the amount of data stored in a single data center can run into hundreds of petabytes and thus necessarily needs to be stored in several different servers. A central problem for such large amounts of data is that of {\it reliable distributed storage}, meaning that the data should be recoverable whenever there is a simultaneous failure of some percentage of the storage servers. A natural solution to this problem is to introduce redundancy by encoding the data via an error-correcting code. The length of this code determines the storage space used while its minimum distance controls the number of simultaneous erasures that can be recovered. However, the recovery of any one erasure may require knowing the values of most other components imposing an excessive communication cost among servers. The theory of \emph{locally recoverable codes} gives us a way to reduce the communication cost by making the recovery local (see e.g. \cites{papailiopoulos2014locally, tamo2014family}).

More precisely, we think of a file as a vector $x\in \F_q^k$ which we encode and store over several storage nodes (servers) via a codeword $c(x)\in \calC$ where $\calC\subseteq \F_q^n$ is a (linear) code of dimension $k$. For simplicity we assume that each of our storage nodes stores exactly one coordinate of $c(x)$. In case of (multiple) node failure, we want to be able to recover the lost information as quickly and efficiently as possible.
In this regard the \emph{locality} of a code plays an important role: it denotes the number of nodes one has to contact for repairing a lost node. We call the set of nodes one has to contact if a given node fails, the locality group of that node and call the collection of locality groups a \emph{locality configuration}.

Furthermore the code $\calC$ is called a \emph{partial maximum distance separable (PMDS)} code if its distinct locality groups are disjoint and it is \emph{maximally recoverable} in the sense that any erasure pattern that is information theoretically correctable is effectively correctable with such a code. PMDS codes are thus objects of great practical importance.

More concretely, a locally repairable code $\calC$ is PMDS with global parameter $s>0$ if its locality configuration $I_1,\dots, I_m$ is a partition of the components and satisfies:
\begin{enumerate}
    \item The restrictions of all codewords to the components indexed by locality set $I_i$ are an MDS code (with length $n_i$ and dimension $k_i$) for $i=1,\dots, m$. 
    \item Any word can be recovered uniquely after erasure of a set $J$ of its components of size $s+\sum_{i} (n_i-k_i)$, when $J$ consists of
    \begin{enumerate}
        \item any $n_i-k_i$ components of the locality set $I_i$ for $i=1,\dots, m$ and
        \item an additional set of any $s$ components.
    \end{enumerate}
    \item $\calC$ has dimension $k=\sum_i k_i -s$.
\end{enumerate}

It is known that PMDS codes exist for any locality configuration if the field size is large enough \cites{ch07,ht19}. Furthermore, some explicit constructions of PMDS codes are known, e.g. \cites{bl13,bl14,bl16,ca17,ch15,gabrys2018constructions, go14,  martinez2019universal}.
However, our knowledge of the structure of PMDS codes remains far from being complete.

Algebraic geometric (AG) codes are a basic source~\cites{goppa1977codes, Pellikaan, MR1186841} of linear codes with interesting structures. Given a (typically irreducible and nonsingular) variety $X$ over $\FF_q$ and sections of a line bundle $\mathcal{L}$ on $X$, such codes are built by evaluating the sections on a given finite subset $\Gamma\subseteq X(\FF_q)$.

In this article we extend algebraic geometric codes to {\it reducible curves} $C$ over $\FF_q$ and use them to define evaluation codes endowed with locality sets defined by the irreducible components of $C$.  We then ask for conditions under which such constructions yield PMDS codes. Our approaches combine techniques from different areas of mathematics, using geometric, combinatorial and probabilistic methods.  The main results in this article are:

\begin{enumerate}
    \item ({\it Geometry of AG PMDS codes}) In Theorem~\ref{pmds_theorem} we give a characterization of the PMDS property for algebraic geometric codes in the language of classical projective geometry. 
    \item({\it Explicit constructions of AG PMDS codes with global parameter $s\in \{1,2\}$}) We use our geometric interpretation of the PMDS property to give simple explicit constructions of PMDS codes with global parameter $s=1$ for all localities. Moreover, in Theorem \ref{thm:PMDSconstruction_s=2} we provide a new construction for $s=2$ and $k_i=2$, which improves the smallest field size obtained by previous explicit constructions. 
    \item ({\it Existence of AG PMDS codes}) In Theorem~\ref{thm: admissible} we prove that there exist geometric PMDS codes for all locality configurations and all global parameters for all sufficiently large field sizes. This method is nonconstructive but leads to an explicit bound on the field size. 
    \item ({\it Randomized construction of AG PMDS codes}) In Sections~\ref{sec:linecodes} and~\ref{sec:probabilistic} we address the lack of explicitness in the previous result. More specifically we specialize our curve $C$ to be an arrangement of lines and analyze the probability that evaluation at a suitably randomized set of points on $C(\FF_q)$ leads to a geometric PMDS code. Our main result is Theorem~\ref{thm:randomcodes} which guarantees that such codes exist whenever $q=O(n^{s})$. Crucially the probabilistic approach is nearly constructive in the sense that for $\epsilon>0$ of our choosing it provides us with a probability distribution on $C(\FF_q)$ for which independent sampling leads to a PMDS code with the desired parameters {\it with probability at least $1-\epsilon$}, 
    \item ({\it Improved probabilistic estimates of field size}) Finally, in Theorem~\ref{thm:alterationcodes} we use the probabilistic method with alterations to obtain estimates for the sizes of fields over which PMDS codes (with localities $(2,\dots,2)$) must exist, improving them to $q=O(n^{s-1})$, a bound which compares favorably with most PMDS existence results in the literature (see Remark~\ref{rmk: comparison}).
\end{enumerate}

The material in the article is organized as follows: Section~\ref{sec: prelims} contains some basic material on coding theory and algebraic geometry over finite fields. Section~\ref{sec:AG} contains the basic construction of codes from reducible curves and a characterization of PMDS codes in the language of projective geometry. Section~\ref{sec:explicit} contains an explicit construction of PMDS codes with global parameter $s\in \{1,2\}$. Section~\ref{sec:existence} proves the existence of algebraic geometric PMDS codes for all localities and all $s$ for sufficiently large fields. Section~\ref{sec:linecodes} focuses on codes constructed from unions of lines and describes the possible obstructions for such codes to be PMDS. Section~\ref{sec:probabilistic} summarizes the key ideas  of the probabilistic method in combinatorics and applies it to prove the probabilistic results described in items $(4)$ and $(5)$ above.
 \medskip
 
{\bf Acknowledgements.} We wish to thank Juan Sebasti\'an Diaz for useful conversations during the completion of this work. A. Neri is funded by \emph{Swiss National Science Foundation}, through grant no. 187711.  M. Velasco is partially supported by research funds from Universidad de los Andes, Facultad de Ciencias, Proyecto  INV-2018-50-1392.

\section{Preliminaries}\label{sec:prelim}

We will use the notation $[n]:= \{1,2,\dots,n\}$ throughout the paper. For a prime power $q$ we let $\mathbb F_q$ denote the finite field with $q$ elements. 

\subsection{Coding theory preliminaries}\label{sec: prelims}
We begin with a brief introduction to the theory of error-correcting codes in the Hamming metric. For a more detailed treatment the reder should refer to~\cites{van_lint}. 

Let $n,k$ be positive integers. By an $[n,k]$ code we mean a linear subspace $\calC \subseteq \FF_q^n$ of dimension $k$. If such a code has minimum Hamming distance $d$ we will call it an $[n,k,d]$ code. By the Singleton bound any $[n,k,d]$ code satisfies $k+d\leq n+1$. The codes achieving the equality are called maximum distance separable (MDS) codes. Such codes are capable of correcting $n-k$ erasures. 

Recall that the group of linear isometries of $\F_q^n$ (i.e., of linear maps that preserve the Hamming distance) consists of componentwise scalings and permutations. More precisely, this group corresponds to $(\F_q^*)^n\rtimes \mathcal S_n$, which acts on $\F_q^n \ni (u_1,\ldots, u_n)$ as
$$(v,\sigma)\cdot (u_1,\ldots, u_n)=(v_1u_{\sigma^{-1}(1)},\ldots, v_nu_{\sigma^{-1}(n)}).$$
Most properties of interest in coding theory are invariant under this group and it is therefore reasonable to introduce the following equivalence relation.

\begin{definition}
We say that two $[n,k]$ codes $\calC$ and $\calC^\prime$ are \emph{equivalent} if there exists a linear isometry of $\FF_q^n$ which maps one onto the other.
\end{definition}

\subsection{PMDS codes}
For $I\subseteq [n]$ we let $\calC_I\subseteq \FF_q^{|I|}$ be the image of the projection of $\calC$ on the coordinates labeled by the indices in $I$.

\begin{definition}\label{pmds_definition}
Let $\calC\subseteq \mathbb F_q^n$ be a linear code of dimension $k$ and let $k_1,\dots, k_m, s$ be positive integers. We say that $\calC$ has \emph{block-locality $(k_1,\dots,k_m)$ with global parameter $s$} if we can write 
$[n]$ as a disjoint union 
\[
[n] = I_1\sqcup \cdots \sqcup I_m
\]
of subsets $I_j$ of cardinalities $n_j$ such that:
\begin{enumerate}
    \item each projection $\calC_i:=\calC_{I_i}$ is a $[n_i,k_i]$ MDS code, and
    \item for any set $J\subset [n]$ such that $|J| = s +  \sum_{i} (n_i-k_i)$ and $|J\cap I_i| \geq n_i-k_i$, the projection map $\calC \rightarrow \calC_{[n]\setminus J}$ is an isomorphism.
\end{enumerate}
We will say that $\calC$ is a \emph{PMDS code}, if $k = \sum_{i} k_i - s$.  For such a code, we call any set $J$ of coordinates as above a \emph{maximal correctable erasure pattern}.
\end{definition}
Equivalently, a code is PMDS if we can correct any $n_i-k_i$ erasures locally in the code $\calC_i$, and any additional $s$ global erasures.  One often wishes to restrict to the homogeneous case wherein $n_i = n/m$ and $k_i = \ell$ for all $i$ and some $\ell$, but we do not necessarily make this assumption. Note that the PMDS property is invariant under code equivalence since it is obviously invariant under coordinate scalings and permutations.

For the PMDS definition to be sensible we need $k\geq \max_i \{k_i\}$. Note that if equality is achieved, with $k=k_i=\ell$ for all $i$, then we recover MDS codes as a special case (see \cite{ht17}*{Proposition 5}). We will therefore assume that $m>1$ and in the homogeneous case that $k>\ell$ throughout.

\subsection{Preliminaries on algebraic geometry}
If $X$ is a variety over $\FF_q$ we let $X(\FF_q)$ be the set of $\F_q$-rational points of $X$. We denote the $k$-dimensional affine (resp. projective) space over $\FF_q$ by $\mathbb{A}^k$ (resp. by $\PP^k$). For a subset $S\subseteq \PP^k$ we denote by $\langle S\rangle$ the projective subspace of $\PP^{k-1}$ spanned by $S$ (i.e., the subvariety defined by the set of linear forms vanishing on $S$). A set of points $T\subseteq \PP^{k-1}$ is in linearly general position if any subset of size $s$ of $T$ spans a projective space of dimension $\min(s-1,k-1)$. A variety $X\subseteq \PP^{k-1}$ is nondegenerate if it is not contained in any hyperplane.

By a rational normal curve $Z$ in $\PP^{k-1}$ we mean a curve projectively equivalent to the image of the $k$-th Veronese morphism $\nu_k: \PP^1\rightarrow \PP^{ k-1}$
given in homogeneous coordinates by all monomials of degree $k-1$, \[\nu_k\left([s:t]\right) = \left[s^{k-1}: s^{k-2}t:\dots: st^{k-2}: t^{k-1}\right].\] 
 By computing a Vandermonde determinant it is easy to see that any set of distinct points lying on a rational normal curve is in linearly general position in $\PP^{k-1}$. Over an algebraically closed field this property characterizes rational normal curves among all irreducible and non-degenerate curves in $\PP^{k-1}$. Rational normal curves can also be characterized as the only non-degenerate irreducible curves of degree $k-1$ in $\PP^{k-1}$~\cite{Harris}*{Proposition 19.9}.

\section{Codes from Reducible Projective Curves}\label{sec:AG}

In this section we give a procedure for constructing evaluation codes endowed with a locality structure from reducible projective curves. To begin, we construct codes from sets of points in projective space.

\subsection{Constructing codes from points in projective space}

\begin{definition}\label{basic_construction}
Let $\Gamma=\{P_1,\ldots,P_n\}\subseteq \PP^{k-1}(\FF_q)$ be a finite set of points.
For each point $P_i$ fix an affine representative $\hat{P_i}\in \mathbb{A}^{k}(\FF_q)$  and let $e$ be the evaluation map
\[
e:\mathcal{L}\rightarrow \F_q^n,\quad e(f) = \left(
f(\hat{P_1}),\ldots,f(\hat{P_n})
\right)\in \F_q^n
\]
where $\mathcal{L} := \left\{a_0X_0 + \cdots +a_{k-1}X_{k-1}\ |\ a_i\in \F_q\right\}.$\\
We define the \emph{algebraic geometric (AG) code} $\calC$ determined by $\Gamma$ to be
\[
\calC = \left\{
e(f)\in \F_q^n\ |\ f\in\mathcal{L}
\right\}.
\]
\end{definition}

\begin{remark} A point $P$ in $\PP^{k-1}(\FF_q)$  has several distinct affine representatives $\hat{P}\in \mathbb{A}^{k}(\FF_q)$ and therefore the previous construction leads to many different possible codes $\mathcal{C}$. We claim that all these choices lead to equivalent codes. This is because if $\hat{P}_j$ and $\hat{P}'_j$ are distinct affine representatives for the point $P_j$ then there exist nonzero scalars $\lambda_j$ such that $\hat{P}'_j=\lambda_j\hat{P}_j$ for every $j=1,\dots, n$. As a result, for any $f\in \mathcal{L}$ the equality $f(\hat{P}'_j)=\lambda_j f(\hat{P}_j)$ holds and therefore the code obtained from evaluation at the representatives $\hat{P}_j'$ is the result of scaling the code obtained from evaluation at the representatives $\hat{P}_j$ by independently scaling the components with the vector $(\lambda_0,\dots, \lambda_n)$. 

Furthermore, an automorphism of $\mathbb P^{k-1}$ acts on $\mathcal{L}$ by a permutation of its elements, leading to a permutation of the words in the code, and hence keeps the code unchanged.
\end{remark}

The previous remark implies that, up to code equivalence, the code $\calC$ defined above is completely determined by the set of points $\Gamma\subseteq \PP^{k-1}$ up to projective automorphisms. It follows that any code property should be interpretable in the language of projective geometry (see~\cite{MR1186841}*{Theorem 1.1.6} for a proof that this is in fact an equivalence). The main result of this section is Theorem~\ref{pmds_theorem} below,  which recasts the PMDS property in the language of projective geometry. 

As a first application of this philosophy, we begin by giving a geometric interpretation to code projections.  
For any set $S\subseteq \Gamma$ we denote by $\calC_S$ the image of the projection of $\calC$ onto the coordinates indexed by the points of $S$. More precisely $\calC_S$ is the image of the composition $\pi_S\circ e: \mathcal{L}\rightarrow \FF_q^{|S|}$ where $\pi_S: \FF_q^n\rightarrow \FF_q^{|S|}$ is the projection onto the coordinates indexed by the points of $S$. The following simple lemma gives a geometric description of the codes $\calC_S$.

\begin{lemma}\label{lem: basic_pt_codes} The composition $\pi_S\circ e$ induces an isomorphism between the space of linear forms in $\langle S\rangle$ and the code $\calC_S$. In particular:
\begin{enumerate}
\item $\langle S\rangle\cong \PP^{t-1}$ if and only if $\calC_S$ has dimension $t$, and
\item $S$ spans the ambient space $\PP^{k-1}$ if and only if $\pi_S:\calC\rightarrow \calC_S$ is an isomorphism.
\end{enumerate}
\end{lemma}
\begin{proof} As in Definition~\ref{basic_construction} let $\mathcal{L}$ be the space of linear forms in $\PP^{k-1}$, let $e: \mathcal{L}\rightarrow \FF_q^{n}$ be the evaluation at the chosen affine representatives of the points of $\Gamma$ and let $\pi_S: \FF_q^n\rightarrow \FF_q^{|S|}$ be the projection. The code $\calC_S$ is by definition the image of $\pi_S\circ e$. The kernel $\mathcal{K}$ of this map consists of the linear forms vanishing identically at all points of $S$ and therefore $\pi_S\circ e$ defines an isomorphism between $\mathcal{L}/\mathcal{K}$ and $\calC_S$. The quotient $\mathcal{L}/\mathcal{K}$ is canonically isomorphic to the space of linear forms on $\langle S\rangle$ proving the initial claim and part (1). 
For part $(2)$ note that $\langle S\rangle=\PP^{k-1}$ if and only if the kernel $\mathcal{K}$ is trivial, and therefore this is equivalent to $\pi_S\circ e$ being an isomorphism between $\mathcal{L}$ and $\calC_S$ which factors through $\calC$.
\end{proof}

\begin{example} (Reed-Solomon codes) Let $Z\subseteq \PP^{k-1}$ be the rational normal curve over $\FF_q$. 
The AG code  $\calC$ defined by $\Gamma:=Z(\FF_q)$ is a (extended) Reed-Solomon code. This code has length $|\Gamma|=q+1$. Moreover, since every set of points of $Z$ is in linearly general position, Lemma~\ref{lem: basic_pt_codes} proves that whenever $q+1\geq k$ the code $\calC$ has dimension $k$ and that the projection of $\calC$ onto any set of $k$ coordinates is an isomorphism, proving that $\calC$ is an MDS code.  
\end{example}

\subsection{Constructing codes from reducible curves}

To construct AG codes with desirable properties, we often specify the sets $\Gamma$ as subsets of other varieties. The following construction, which uses reducible curves, will be our source for constructing PMDS codes.

\begin{definition}\label{main_construction}
Let $C_1,\ldots,C_m\subseteq \PP^{k-1}$ be distinct irreducible curves over $\FF_q$. Define $C := C_1\cup\cdots\cup C_m$
and let $\Gamma=\{P_1,\ldots,P_n\}\subseteq C(\F_q)$ be a given set of points, each lying in at most one of the $C_i$, and define a partition $\Gamma:=\Gamma_1\sqcup \cdots \sqcup \Gamma_m$ into $m$ disjoint subsets of cardinalities $n_i$, determined by the irreducible components $C_i$ of $C$ via $\Gamma_i:=\Gamma\cap C_i(\FF_q)$. 
Let $\calC$ be the AG code defined by $\Gamma\subseteq \PP^{k-1}$ as in Definition~\ref{basic_construction}. 
For $i = 1,\ldots,m$, let $\calC_i:=\calC_{\Gamma_i}$ denote the image of the projection of $\calC$ onto the coordinates in $[n]$ corresponding to the points of $\Gamma_i$. We will refer to the $\calC_i$ as the \emph{local codes} of $\calC$. 
\end{definition}

Assume moreover that we are given positive integers $k_1,\dots, k_m,k$ and $s$ which satisfy the inequalities $n_i\geq k_i$ for every $i$ and the equality $s+k=\sum_{i} k_i$.

\begin{definition}\label{proj_pmds_definition} We call a set $S\subseteq \Gamma\subseteq \PP^{k-1}$ an {\it evaluation set} if  $|S\cap \Gamma_i|\leq k_i$ for $i=1,\dots, m$. We say that $\Gamma $ is {\it admissible} if the following two conditions hold: 
\begin{enumerate}
\item The projective subspace $\langle\Gamma_i\rangle$ has dimension $k_i-1$ and the points of $\Gamma_i$ are in linearly general position in  $\langle\Gamma_i\rangle$. 
\item Every evaluation set $S\subseteq \Gamma$ of size $k=\sum_{i} k_i -s$  spans $\PP^{k-1}$. 
\end{enumerate}
\end{definition}

\begin{theorem}\label{pmds_theorem}
The following statements are equivalent:
\begin{enumerate}
\item The code $\calC$ from Definition~\ref{main_construction} is a PMDS code with blocks given by the $\Gamma_i$, block locality $(k_1,\dots, k_m)$ and global parameter $s$. 
\item $\Gamma$ is an admissible set with respect to the partition $\Gamma_1,\dots, \Gamma_m$.
\end{enumerate}
 \end{theorem}
\begin{proof} We will prove the claim by showing that property $(1)$ (resp. $(2)$) in Definition~\ref{pmds_definition} and property $(1)$ (resp. $(2)$) of Definition~\ref{proj_pmds_definition} are  equivalent. By construction the code $\calC_{\Gamma_i}$ has length $n_i$. By Lemma~\ref{lem: basic_pt_codes} the code $\calC_{\Gamma_i}$  has dimension $k_i$ if and only if the space $\langle \Gamma_i\rangle$ is isomorphic to $\PP^{k_i-1}$. Furthermore $\calC_{\Gamma_i}$ is a $[n_i,k_i]$ MDS code if and only if the projection $\pi_T: \calC_{\Gamma_i}\rightarrow \calC_T$ is an isomorphism for every set $T\subseteq \Gamma_i$ with $|T|=k_i$. By Lemma~\ref{lem: basic_pt_codes} this condition is equivalent to the fact that the points of $\Gamma_i$ are in linearly general position in $\langle \Gamma_i\rangle$. Finally let $J\subseteq [n]$ be a set with $|J\cap \Gamma_i|\geq n_i-k_i$ and $|J|=s+\sum_i (n_i-k_i)$ and let $S$ be the complement of $J$. It follows that $S$ is admissible for $\Gamma$ and has cardinality $k$.
By Lemma~\ref{lem: basic_pt_codes} $S$ spans $\PP^{k-1}$ if and only if the projection $\pi_S:\calC\rightarrow \calC_S$ is an isomorphism. 
\end{proof}

The following two examples illustrate our construction.
 
\begin{example}[Explicit construction of a simple PMDS AG code]\label{simple_pmds_example}
 We will construct an AG code using $m = 2$ components.  Let $C_1$ and $C_2$ be the $x$- and $y$-axes, respectively, in the projective plane $\PP^2$:
  \[
  \begin{aligned}
  C_1 &:= \left\langle
  [1:0:1], [0:0:1]
  \right\rangle \\
  C_2 &:= \left\langle
  [0:1:1], [0:0:1]
  \right\rangle
  \end{aligned}
  \]
  and let $C = C_1\cup C_2$.  Define
  \[
  \Gamma = \left(C_1(\F_q)\cup C_2(\F_q)\right)\setminus\{[0:0:1]\}
  \]
  which contains $n = 2q$ evaluation points. The resulting AG code $\mathcal{C}$ clearly has $m = 2$ components, each giving an MDS local code $\mathcal{C}_i$ with parameters $[n_i,k_i] = [q,2]$.  This last statement is immediate because any set of distinct points of cardinality at least two on a line spans it and is in linearly general position. Moreover, the whole code $\mathcal{C}$ has dimension $k = 3$, since the set $\Gamma$ clearly spans all of $\PP^2$.
  
  Lastly, we claim that $\mathcal{C}$ is a PMDS code with global parameter $s = k_1 + k_2 - k=2+2-3=1$.  Indeed, the complement $S$ of any maximal erasure pattern $J$ as in Definition \ref{pmds_definition} corresponds to three points on $C$, with each of the lines $C_i$ containing at most two of these points.  Such a set will always span $\PP^2$, hence any such $S$ is an information set of $\mathcal{C}$ and maximal erasure patterns are always correctable.  This proves that $\mathcal{C}$ is a PMDS code.
\end{example}

\begin{example}[Non-example of a PMDS AG code]\label{fail_example}
Let us show by example that a na\"ive generalization of Example \ref{simple_pmds_example} to other parameters results in sets $S\subseteq\Gamma$ of evaluation points such that $|S\cap C_i| = k_i$ and $|S|=k$, but $\langle S\rangle \subsetneq \PP^{k-1}$.  We choose $m = 2$, $k_1 = k_2 = 3$, and $s = 2$.  Thus $k = 6-2 = 4$, and the reducible curve $C = C_1 \cup C_2$ is a union of two conics in $\PP^3$. Choose any subset $S = \{P_1,P_2,Q_1,Q_2\}\subseteq\Gamma$ such that $P_i\in C_1$ and $Q_i\in C_2$.  For the resulting AG code $\mathcal{C}$ to be PMDS, it is necessary that $S$ spans all of $\PP^3$.  Suppose that $P_1, P_2, Q_1$ are not collinear (a necessary condition for $\mathcal{C}$ to be PMDS) and therefore span a hyperplane $H$.  The intersection of $H$ and $C_2$ consists of two points, one of which is $Q_1$ and the other of which is some other point $\widetilde{Q}_2$ also defined over $\F_q$.  If $Q_2 = \widetilde{Q}_2$ then $S$ only spans a hyperplane, and the code cannot correct the corresponding erasure pattern.
\end{example}

This last example shows that a more careful choice of $\Gamma$ is necessary to guarantee the PMDS property.  More precisely, our constructions in the following sections will choose the evaluation points so that no such co-hyperplanar critical evaluation sets exist as subsets of $\Gamma$.  For example, for all subsets $\{P_1,P_2,Q_1\}\subseteq \Gamma$ as in Example \ref{fail_example}, we must choose $\Gamma$ so that (i) this triple is not collinear, and (ii) $\widetilde{Q}_2\not\in\Gamma$.  A much sparser subset of the $\F_q$-points of $C$ is necessary to achieve this goal.

\begin{remark}
It is common in coding theory to construct evaluation codes from algebraic curves in a more abstract setting: We are given a curve $C$ (typically irreducible and non-singular), a vector space $V$ of sections of a line bundle  $\mathcal{L}$ on $C$ and a finite set $\Gamma\subseteq C(\FF_q)$. By embedding $C$ in projective space via the morphism $\phi$ specified by the sections in $V$ (resolving indeterminacies if necessary) we can think of $C$ or rather of its image $\phi(C)$ as a curve in $\PP(V^*)$ and apply the construction above with the set $\phi(\Gamma)$. The concrete projective approach above is therefore more general than the abstract approach.
\end{remark}

\section{Explicit Constructions of Algebraic Geometric PMDS codes with Global Parameter One and Two} \label{sec:explicit} 

In this section we focus on explicit constructions of geometric PMDS codes for $s=1,2$. We begin with $s=1$ for any choice of localities $(k_1,\dots, k_m)$. We hope its simplicity will convince the reader of the usefulness of the projective viewpoint when constructing PMDS codes. Given $(k_1,\dots, k_m)$ we let $k:=\sum_i k_i -1$ and construct a reducible curve $C\subseteq \PP^{k-1}$ which is a disjoint union of $m$ rational normal curves of degrees $k_i-1$. Furthermore we will construct a finite set $\Gamma\subseteq C(\FF_q)$ and show that it is admissible concluding, via Theorem~\ref{pmds_theorem}, that the corresponding code $\calC$ is indeed PDMS. More precisely, let $R=\sum_i k_i$ and let $P_1,\dots, P_{R}$ be a set of points in linearly general position in $\PP^{R-2}(\FF_q)$. Split the points into $m$ disjoint subsets of sizes $k_1,\dots, k_m$ and let $\Lambda_j$ be the projective subspaces (of dimension $k_j-1$) spanned by the subsets. Note that for each $j\in\{1,\ldots,m\}$ the subspace $H_j=\langle \Lambda_i: i\neq j\rangle$ is a projective space of dimension $R-k_j-1$ that intersects $\Lambda_j$ in exactly one point, which we denote by $Q_j$. For each $j$ let $C_j$ be a rational normal curve over $\FF_q$ of degree $k_j-1$ in $\Lambda_j$ passing through $Q_j$ and through those $P_i's$ which are contained in $\Lambda_j$  (such a curve exists because there are $k_j+1$ such points and they are linearly independent). Finally let $\Gamma_j:=C_j(\FF_q)\setminus \{Q_j\}$, $\Gamma:=\Gamma_1\sqcup\cdots\sqcup\Gamma_m$ and $\calC$ be the AG code defined by $\Gamma$.

\begin{theorem}
  The set $\Gamma$ is admissible  and therefore the code $\calC$ is a PMDS code with localities $(k_1,\dots, k_m)$ and $s=1$.
\end{theorem}

\begin{proof} 
  Clearly, the points in every $\Gamma_i$ are in general position in $\langle \Gamma_i \rangle $, since $\Gamma_i$ is a subset of a rational normal curve which spans $\Lambda_i$. We only need to show that every evaluation set  $S\subseteq \Gamma$ of size $R-1$ spans $\PP^{R-2}$. By definition of evaluation set, we have that there exists a $j\in \{1,\ldots, m\}$ such that $|S\cap \Gamma_j|=k_j-1$, while $|S\cap \Gamma_i|=k_i$ for every $i \neq j$. Therefore $\langle S\cap \Gamma_i\rangle=\Lambda_i$ for every $i\neq j$. Let $\{P_1',\ldots ,P_{r_j-1}'\}=S\cap \Gamma_j$, and call $\Lambda_j'=\langle S\cap \Gamma_j\rangle$, which has dimension $k_j-2$.
  Hence, we have $\langle S\rangle=\langle H_j, \Lambda_j'\rangle$. In order to prove that $\langle S\rangle=\PP^{R-2}$ it is equivalent to show that $H_j\cap \Lambda_j'=\emptyset$, by a dimension argument. Suppose by contradiction that $H_j\cap \Lambda_j'\neq \emptyset$. Since $\Lambda_j'\subset \Lambda_j$ and by construction $H_j\cap\Lambda_j=\{Q_j\}$, then this implies that $Q_j \in \Lambda_j'$. Therefore, $P_1',\ldots,P_{r_j-1}, Q_j$ lie on the same $k_i-2$-dimensional subspace. However, this is not possible by construction, since the points are all distinct and lie on the rational normal curve $C_j$. This concludes the proof. 
\end{proof}

\begin{remark}
This construction slightly generalizes the one given in \cite{ch15}, where the authors used extended Reed-Solomon codes as local codes. Their construction indeed coincides with ours when $k_1=\ldots=k_m$. 
Moreover, it is a special case of the characterization given in \cite{ht17}, where it was shown that every PMDS code with  global parameter $s=1$ is constructed using MDS codes as local codes.
\end{remark}

We conclude this section by providing an explicit construction of algebraic geometric PMDS codes with global parameter $s=2$ and localities $k_i=2$. This is the first known explicit construction for such codes over finite fields of cardinality $q\geq n-1$. We remark that, in the homogeneous case where also  $n_i=n/m$, the construction in \cite{bl14} requires a field size $q>n$.

We will restrict our attention to codes constructed from evaluation sets $\Gamma \subseteq C(\FF_q)$ where $C$ is a reducible curve in $\PP^{k-1}$ with the property that its components are $m$ distinct generic {\it lines}.  More precisely, we assume that $m$ is a given positive integer, that $k_i=2$ for $i=1,\dots, m$ and that $s=2$, so the equality $k=\sum_i k_i - s$ becomes $k=2m-2$. Assume $q\geq 2m-1$, let $Z$ be the rational normal curve in $\PP^{2m-3}$ and choose $2m$ distinct points $P_1,Q_1,P_2,Q_2,\dots, P_m,Q_m$ in $Z(\FF_q)$. Define the lines  $L_i:=\langle P_i,Q_i\rangle$ for $i=1,\dots, m$ and let $C=\bigcup_{i} L_i$. We will show that an appropriate choice of $\Gamma\subseteq C(\FF_q)$ is admissible and conclude, via Theorem~\ref{pmds_theorem}, that the corresponding code $\calC$ is PDMS with localities $(2,\dots, 2)$ and $s=2$.

The key point of the construction is understanding the structure of the obstructions to admissibility, which we can do completely explicitly for $s=2$. To this end, for every $i\neq j$, define the maps  $f_{i,j}:L_i\rightarrow L_j$ by 
\begin{equation}\label{eq:fij}
f_{i,j}(P) := \langle L_1,\ldots,\hat{L}_i,\ldots,\hat{L}_j,\ldots,L_m,P\rangle\cap L_j, \end{equation}
where the notation $\hat{L}_i$ means that  $L_i$ is not taken in the spanning set.  
Furthermore, define $f_{i,i}:=\mathrm{id}$ for every $i \in \{1,\ldots,m\}$. 

Note that $f_{i,j}$ is well-defined, since  any $m-2$ lines and an additional single point on one of the remaining lines always span a hyperplane, which intersects the last line in exactly one point.

\begin{lemma} With the notation above, the following hold.
\begin{enumerate}
\item $f_{j,\ell}\circ f_{i, j}=f_{i, \ell}$ for every $i,j,\ell\in \{1,\ldots, m\}$. 
\item $f_{i,j}$ is bijective for every $i,j \in \{1,\ldots, m\}$.
\end{enumerate}
\end{lemma}

\begin{proof}~\begin{enumerate} 
\item 
If $i=j$ or $j=\ell$, then the statement is trivial. If $i=\ell$ and $i\neq j$ then define the space $Y:=\langle L_r \mid r \notin \{i,j\} \rangle$. Fix a point $P \in L_i$ and call $Q:=f_{i,j}(P)$, that is $\langle Y,P\rangle \cap L_j=\{Q\}$. Hence $\langle Y,P\rangle =\langle Y, P, Q \rangle =\langle Y, Q\rangle$, since $Q \notin Y$. This implies that $P \in \langle Y,Q\rangle$, and thus $f_{j,i}(Q)=P$. 

Let now $i,j,\ell\in\{1,\ldots,m\}$ be pairwise distinct. Define the space $X:=\langle L_r \mid r \notin \{i,j,\ell\} \rangle$. Fix a point $P \in L_i$ and call $Q:=f_{i,j}(P)$ and $R:=f_{j,\ell}(Q)$. This means that $\langle X,L_{\ell},P\rangle \cap L_j=\{Q\}$ and 
$\langle X,L_{i},Q\rangle \cap L_\ell=\{R\}$. We want to show that $f_{i,\ell}(P)=R$. Consider now the two spaces $\langle X,L_{\ell},P\rangle, \langle X,L_{i},Q\rangle$ and let $\Lambda$ be their intersection. The two spaces are distinct hyperplanes and hence $\dim \Lambda=2m-5$. Moreover, it is easy to see that $P,Q,R \in \Lambda$. 
Consider now the space $\langle X, P, Q \rangle$. Clearly, $P, Q\notin X$, so it has dimension at least $2m-6$. Moreover, suppose that $P\in \langle X, Q\rangle$. Then we would have $L_i \cap \langle X,L_j \rangle \neq \emptyset$, which would produce a dependence among the points$\{P_1,Q_1\ldots,P_{m}, Q_m\}\setminus\{P_\ell,Q_\ell\}$. This contradicts the fact that the points $P_1, Q_1\ldots,P_{m},Q_m$ are in general positions. Therefore,
we can deduce that $\langle X, P, Q \rangle$ has dimension $2m-5$. Hence $\langle X,P,Q \rangle =\Lambda$. Take now the space $\langle X, L_j, P\rangle$. We have that $\Lambda \subseteq \langle X, L_j, P\rangle$, and we conclude that $R \in \langle X, L_j, P\rangle$ and thus $f_{i,\ell}(P)=R.$

\item  
Follows immediately from $(1)$ with $i=\ell$. 
\end{enumerate}
\end{proof}

 The maps $f_{i,j}$ allow us to define an equivalence relation on $C$ by saying 
 that two points $P\in L_i$, $Q\in L_j$ are equivalent if and only if $f_{i,j}(P)=Q$.
 By the previous lemma, the equivalence class of a point $P\in L_i(\FF_q)$ is given by the set $S_P:=\{f_{i,j}(P) \mid 1\leq j \leq m \}$. These sets give us a partition of $C(\FF_q)$ into $q+1$ disjoint sets of size $m$ and furthermore, for every $j$, the points in $L_j(\FF_q)$ are a system of distinct representatives for the equivalence relation (i.e., the points of $L_j(\FF_q)$ parametrize the equivalence classes). 
 
 Let $\Gamma\subseteq C(\FF_q)\subseteq \PP^{2m-2}$ be a set containing at least two points from each $L_i(\FF_q)$ and at most one element in each equivalence class, and let $\Gamma_j:=\Gamma\cap L_j(\FF_q)$. Note that $\Gamma$ exists because we are assuming $q+1\geq 2m$. We are now in a position to prove the main result of this section.
 
\begin{theorem}\label{thm:PMDSconstruction_s=2}
 Any set $\Gamma=\Gamma_1\sqcup \Gamma_2\sqcup \dots \sqcup \Gamma_m$ defined as above is admissible  and therefore the AG code $\calC$ defined by $\Gamma$ is a PMDS code with localities $(2,\dots, 2)$ and $s=2$. 
\end{theorem}

\begin{proof}
Since we chose the  sets $\Gamma_i$ to have cardinality at least $2$, we have $\langle \Gamma_i \rangle=L_i$. Let $S\subseteq \Gamma$ be an evaluation set of size $k=2m-2$. There are two possibilities for the intersection of $S$ with the $\Gamma_i$'s:

\noindent Case I: There exists $i \in \{1,\ldots, m\}$ such that $S \cap \Gamma_i=\emptyset$ and $|S\cap \Gamma_j|=2$ for every $j \neq i$. Hence,
$$\langle S \rangle =\langle L_j \mid j\neq i \rangle=\langle \{P_j,Q_j \mid j \neq i\} \rangle=\PP^{2m-3},$$
where the last equality follows from the fact that the points $\{P_1,Q_1,\ldots,P_m,Q_m\}$ are in general position.

\noindent{Case II:} There exists two distinct integers $i,j\in\{1,\ldots,m\}$ such that $|S\cap \Gamma_i|=|S\cap \Gamma_j|=1$ and $|S\cap \Gamma_\ell|=2 $ for every $\ell \notin \{i,j\}$. Let $P$ be the point in $S\cap \Gamma_i$ and $Q$ be the point in $S \cap \Gamma_j$. Moreover, we define $X:=\langle L_\ell \mid \ell \notin \{i,j\}\rangle$, which by genericity of the points $P_i$'s has dimension $2m-5$ and is contained in $\langle S \rangle$. Moreover, again by genericity, none of the points $P,Q$ belongs to $X$. Hence $\langle S \rangle$ has dimension at least $2m-4$, and has dimension exactly $2m-4$ if and only if $Q \in \langle X,P\rangle$. By definition of the map $f_{i,j}$, this is true if and only if $f_{i,j}(P)=Q$, that is, $P$ and $Q$ belong to the same equivalence class. However, this is not possible, since we only selected at most one point in each equivalence class for constructing the set $\Gamma$. Thus, $\langle S\rangle =\PP^{2m-3}$.

Therefore, $\Gamma$ is admissible and $\calC$ is a PMDS code by Theorem~\ref{pmds_theorem}.
\end{proof}
 
\begin{example}
We illustrate the result of Theorem \ref{thm:PMDSconstruction_s=2} by constructing a $[20,6]$ PMDS code over $\FF_{19}$, with $m=4$, $n_i=5$ and $k_i=2$ for $i\in\{1,2,3,4\}$. We choose the following eight points in $\PP^5$:
$$\begin{array}{rlrl} &P_1=  [1:0:0:0:0:0] & &Q_1 = [0:1:0:0:0:0], \\ 
 &P_2= [0:0:1:0:0:0], & &Q_2 = [0:0:0:1:0:0], \\
 &P_3= [0:0:0:0:1:0], & &Q_3 =[0:0:0:0:0:1], \\ 
 &P_4=  [1:1:1:1:1:1], & &Q_4= [1:2:4:8:16:13]. \end{array}$$
 Then we define the lines $L_i=\langle P_{i},Q_{i}\rangle$, for $i\in\{1,2,3,4\}$.
 The first line $L_1$ is then given by
 $$ L_1=\{ R_x=[1:x:0:0:0:0] \mid x \in \F_{19} \}\cup \{Q_1\}.$$
 We can compute $f_{1,4}(Q_1)=Q_4-P_4$. Moreover, for the remaining points, define $\lambda_x:=\frac{x-1}{2-x}$ for any $x\neq 2$. We have
 $$f_{1,4}(R_x)=\begin{cases} Q_4 & \mbox{ if } x=2, \\ 
 P_4+\lambda_xQ_4 & \mbox{ if } x\neq 2. \end{cases}$$
 One can also compute $f_{1,3}$, obtaining $f_{1,3}(Q_1)=15P_3+12Q_3$ and 
  $$f_{1,3}(R_x)=\begin{cases} P_3+2Q_3 & \mbox{ if } x=2, \\ 
 P_3+\frac{(1+13\lambda_x)}{(1+16\lambda_x)}Q_3 & \mbox{ if } x\neq 2. \end{cases}$$
 Finally, we have $f_{1,2}(Q_1)=3P_2+7Q_2$ and 
   $$f_{1,2}(R_x)=\begin{cases} P_2+2Q_2 & \mbox{ if } x=2, \\ 
 P_2+\frac{(1+8\lambda_x)}{(1+4\lambda_x)}Q_2 & \mbox{ if } x\neq 2. \end{cases}$$
 This gives us the equivalence classes
 \begin{align*}
     S_{Q_1} &=\left\{Q_1, 3P_2+7Q_2, 15P_3+12Q_3, P_4-Q_4 \right\}, \\
     S_{R_x} &=\begin{cases}\left\{R_x,P_2+2Q_2,P_3+2Q_3, Q_4 \right\} & \mbox{ if } x=2, \\ 
     \left\{R_x,P_2+\frac{(1+8\lambda_x)}{(1+4\lambda_x)}Q_2,P_3+\frac{(1+13\lambda_x)}{(1+16\lambda_x)}Q_3,P_4+\lambda_xQ_4  \right\} & \mbox{ if } x\neq2. \\ \end{cases} 
 \end{align*}
 Now, we select a point from each set $S_{R_x}$. For instance, we do the following choice: for each set $S_{R_x}$, we take the point in $L_i\cap S_{R_x}$ if and only if $x \equiv i \mod 4$, and we then select the point $15P_3+12Q_3$ from $S_{Q_1}$ which belongs to $L_3$. This produces the $[20,6]$ code whose generator matrix is
 \small{$$\left( \begingroup 
\setlength\arraycolsep{2pt}\begin{array}{ccccc|ccccc|ccccc|ccccc} 
 1 & 1 & 1 &  1 &  1 &  0 & 0 & 0 &  0 &  0 &  0 &  0 &  0 &  0 &  0 & 11 & 12 & 16 &  5 & 10\\
 1 & 5 & 9 & 13 & 17 &  0 & 0 & 0 &  0 &  0 &  0 &  0 &  0 &  0 &  0 &  2 &  4 & 12 &  9 &  0\\
 0 & 0 & 0 &  0 &  0 &  1 & 1 & 1 &  1 &  1 &  0 &  0 &  0 &  0 &  0 &  3 &  7 &  4 & 17 & 18\\
 0 & 0 & 0 &  0 &  0 &  2 & 7 & 5 & 17 & 14 &  0 &  0 &  0 &  0 &  0 &  5 & 13 &  7 & 14 & 16\\
 0 & 0 & 0 &  0 &  0 &  0 & 0 & 0 &  0 &  0 &  1 &  1 &  1 &  1 &  1 &  9 &  6 & 13 &  8 & 12 \\
 0 & 0 & 0 &  0 &  0 &  0 & 0 & 0 &  0 &  0 & 10 & 15 & 12 & 18 & 16 & 17 & 11 &  6 & 15 &  4\\
 \end{array} \endgroup \right).$$}
 Observe that if we have fixed $m=4$ and $k_1=k_2=k_3=k_4=2$, then starting with the sets $S_{Q_1}$ and $S_{R_x}$ for $x \in \FF_{19}$ computed above, we can construct PMDS codes of any length $n\leq 20$, where the local codes are $[n_i,2]$ codes over $\F_{19}$.
\end{example}

\begin{remark}
It is tempting to try a similar approach for constructing PMDS codes with global parameter $s>2$ starting from a reducible curve composed of $m$ generic lines. However, the obstructions quickly become difficult to manage because one has to explicitly describe several distinct combinatorial types of obstructions, the number of types increasing with $s$. In Sections~\ref{sec:linecodes} and ~\ref{sec:probabilistic} we abandon this approach in favor of methods from probability theory aiming to quantify the relative sizes of such obstructions in order to obtain results for all $s>2$.
\end{remark}

\section{Existence of Algebraic Geometric PMDS Codes} \label{sec:existence}

Assume we are given positive integers $m,s$, $k_1,\dots, k_m$ and $k$ satisfying $k_i<k$ and $k+s = \sum_i k_i$. In this section we prove the existence of algebraic geometric PMDS codes with arbitrary localities $(k_1,\dots, k_m)$ and global parameter $s$ over $\FF_q$ {\it for all sufficiently large field sizes $q$}. To this end let $Z$ be the rational normal curve of degree $k-1$ in $\PP^{k-1}$. If $q+1 > \sum_i k_i$ then $Z(\FF_q)$ contains a set $\Gamma_0$ consisting of $\sum_i k_i$ distinct points in $Z(\FF_q)$. Split $\Gamma_0$ into $m$ disjoint subsets $\Gamma_0^{i}$ of size $k_i$, $i=1,\dots, m$. Since $\Gamma_0$ consists of distinct points in the rational normal curve $Z$, the points of $\Gamma_0$ are in linearly general position and in particular the projective subspace $W_i:=\langle \Gamma_0^i\rangle$ has dimension $k_i-1$. Let $C_i$ be a rational normal curve of degree $k_i-1$ in $W_i$ over $\FF_q$ containing the points of $\Gamma_0^{i}$ and let $C:=\bigcup_{i} C_i$. By construction, $C_i \cap C_j = \emptyset$ if $i \neq j$. 
The reducible curve $C$ will be our main tool for constructing AG PMDS codes as in Definition~\ref{main_construction}.

\begin{remark}
Note that the set $\Gamma_0$ consists of $\sum_i k_i$ points, including exactly $k_i$ in $C_i$. Since $\Gamma_0\subseteq Z\subseteq \PP^{k-1}$ the points of $\Gamma_0$ are in linearly general position. It follows that every subset $S\subseteq \Gamma_0$ is an evaluation set and that every subset of size $k$ of $\Gamma_0$ spans $\PP^{k-1}$. We conclude that $\Gamma_0$ is admissible for every $0\leq s < \sum_i k_i$.
\end{remark}

The number of erasures that the PMDS code of an admissible $\Gamma$ can recover is precisely $|\Gamma|+s-\sum_i k_i$ and therefore we would like $\Gamma$ to be both admissible and as large as possible. We will show that, whenever the field size $q$ is sufficiently large, it is possible to add a point to an admissible set and obtain a bigger, but still admissible set. 

We say that a component $C_i$ of $C$ is \emph{selected} by an evaluation set $S$ whenever $|S\cap C_i(\FF_q)|=k_i$. The key to the construction is the following Lemma

\begin{lemma}\label{lem:admissible_hyperplane} If $\Gamma$ is admissible then every evaluation set $S\subseteq \Gamma$ of size $k-1$ spans a hyperplane in $\PP^{k-1}$ which does not contain any component not selected by $S$. 
\end{lemma} 
\begin{proof} Let $S\subseteq \Gamma$ be an evaluation set of size $k-1$. Since $\sum_i k_i -(s+1)=k-1$ there exists a component $\Gamma_j$ such that $|S\cap \Gamma_j|<k_j$ and in particular there is a point $x_j\in \Gamma_j$ which is not in $S$.
If $\langle S\rangle$ has codimension at least two then by adding $x_j$ to $S$ we obtain an evaluation set $S'\subseteq \Gamma$ of cardinality $k$ which does not span $\PP^{k-1}$, contradicting the admissibility of $\Gamma$. It follows that $\langle S\rangle$ spans a hyperplane in $\PP^{k-1}$. Suppose that a component $C_j$ not selected by $S$ satisfies $C_j\subseteq \langle S\rangle$. Since $C_j$ is not selected by $S$ the strict inequality $|S\cap C_j|<k_j$ holds and therefore there exists a point $x_j\in C_j$ with $x_j\not \in S$. It follows that the set $S':=S\cup\{x_j\}$ has size $k$ and $\langle S\rangle = \langle S'\rangle$, so $S'$ is an evaluation set of size $k$ which is not linearly independent contradicting the admissibility of $\Gamma$. It follows that the span of $S$ does not contain any component of $C$ not selected by $S$, as claimed.
\end{proof}

\begin{theorem}\label{thm: admissible} Let $\Gamma\subseteq C$ be an admissible set containing $\Gamma_0$. For all sufficiently large $q$ there exists a point $x\in C(\FF_q)$ such that $\Gamma\cup \{x\}$ is admissible.\end{theorem}

\begin{proof} Since $\Gamma$ is finite, it has a finite set of evaluation subsets $S$ of size $k-1$. By  Lemma \ref{lem:admissible_hyperplane}, every such subset has the property that $\langle S\rangle$ is a hyperplane in $\PP^{k-1}$ and this hyperplane intersects every curve $C_j$ not selected by $S$ at a finite set of at most $k_j$ points. If we were to add to $\Gamma$ any such point $z$, it would immediately become inadmissible since it would contain the linearly dependent evaluation set $S\cup\{z\}$. Therefore we want to forbid those choices. Every such evaluation set $S$ forbids a set of at most $\sum_i k_i$ points of $C$ (a bound independent of the field size $q$). In particular, the number of points forbidden by some $S$ is at most the product $h$ of the number of evaluation sets of $\Gamma$ of size $k-1$ times $\sum_i k_i$. If $q>h$ then every component of $C$ has at least one non-forbidden point.

  Choose any component $C_j$ and any such point $z\in C_j$ and let $\Gamma':=\Gamma\cup \{z\}$. We claim that $\Gamma'$ is admissible. Since $\Gamma'\supseteq \Gamma_0$ it contains at least $k_i$ points in $C_i$ for $i=1,\dots, m$. If $S\subseteq \Gamma'$ is any evaluation set of size $k$ then either $S\subseteq \Gamma$ and therefore $S$ spans $\PP^{k-1}$ since $\Gamma$ is admissible or $z\in S$. In this case $T:=S\setminus \{z\}$ is an evaluation subset of $\Gamma$ and therefore spans a hyperplane $H$. The intersection of $H$ with the non-selected components is a finite set of forbidden points and therefore this set does not contain $z$. It follows that $\langle T\rangle$ is a proper subset of $\langle S \rangle$ and thus $S$ spans $\PP^{k-1}$ as claimed.
\end{proof}

The previous argument shows that one can grow $\Gamma$ by adding points one at a time in any component we want whenever the field is sufficiently large.

\begin{remark} The previous argument also gives us an effective bound in the field size since the 
following inequality holds
\[h\leq \binom{|\Gamma|}{k-1}\sum_i k_i\]
where $h$ is defined as in the previous proof. Let $n$ be the length of the code after completing the process; that is, $n$ is one more than the size of the largest (last) subset $\Gamma$ that we need to augment. Then the process is possible if
\[q > \binom{n-1}{k-1}\sum_i k_i
=\binom{n-1}{\sum_i k_i  - s - 1}\sum_i k_i.\]
%
It was proved in \cite{ch07} that for $q>\binom{n-1}{k-1}$ PMDS codes always exist (for any choice of $k_i$'s). This result was also improved with a similar approach in \cite{ht19}, taking into account also the values of the $n_i$'s and $k_i$'s. Note however that these results are not directly comparable to ours since we are proving the stronger claim that for such field sizes there exist {\it algebraic geometric} PMDS codes.
\end{remark}

\section{PMDS obstructions from matroid theory.} \label{sec:linecodes}

In this section we will restrict our attention to codes constructed from evaluations at a set of points $\Gamma \subseteq C(\FF_q)$ where $C$ is a reducible curve in $\PP^{k-1}$ with the property that its components are lines. To emphasize this difference we denote the components of $C$ with $L_1,\dots, L_m$ and not with $C_1,\dots C_m$ as in the previous sections of the article. More precisely, we assume that $m,s$ are given, that $k_i=2$ for $i=1,\dots, m$ and that $s$ satisfies $0\leq s<2m$ and let $k:=2m-s$. We let $Z$ be the rational normal curve in $\PP^{k-1}$, choose $2m$ distinct points $P_1,Q_1,P_2,Q_2,\dots, P_m,Q_m$ in $Z(\FF_q)$ and define the lines  $L_i:=\langle P_i,Q_i\rangle$ for $i=1,\dots, m$ and the set $\Gamma:=C(\FF_q)$. For brevity we will refer to this construction by letting $C$ be a reducible curve composed of $m$ generic lines in $\PP^{k-1}$.

The main results of this section are Corollary~\ref{cor:PMDScriterion} which identifies the obstructions for the AG code corresponding to $\Gamma$ to be a PMDS code and Corollary~\ref{cor:ccbound} which gives us an estimate for the total number of these obstructions. Both of these results will be used in the next section to derive our improved estimates on field size.

These obstructions are most easily described in the language of matroids. We think of the set $C(\FF_q)$ as a matroid, that is as a set together with a collection of distinguished {\it dependent} subsets, by saying that $S\subseteq C(\FF_q)$ is dependent if the projective space $\langle S\rangle\subseteq \PP^{k-1}$ has dimension at most $|S|-2$. Notice that a matroid point of view was already used in \cites{tamo2016optimal, westerback2016combinatorics} for the study of locally repairable codes.

\begin{definition} A \emph{circuit} $S\subseteq C(\FF_q)$ is a dependent set which is minimal with respect to inclusion. We distinguish two special kinds of circuits:
\begin{enumerate}
    \item A circuit is \emph{trivial} if it contains at least three points which belong to one of the lines.
    \item A circuit is \emph{crossing} if it contains at most one point in each of the lines.
\end{enumerate}
For a set $S\subseteq C(\FF_q)$ we define the \emph{range} of $S$ to be the set of indices $j$ of those $L_j$ for which $L_j\cap S\neq \emptyset$. In particular the cardinality of every crossing circuit equals the cardinality of its range. 
\end{definition}

\begin{lemma} \label{lem:existcc} For any nontrivial circuit $C$ there exists a crossing circuit $D$ such that the following statements hold:
  \begin{enumerate}
    \item $\textup{range}(D) = \textup{range}(C)$.
    \item if $C$ contains only a single point $q_j$ on some line $L_j$, then $q_j \in D$.
   \end{enumerate}
Moreover, the cardinality $u$ of every crossing circuit satisfies $\lceil\frac{k+1}{2}\rceil\leq u$.
\end{lemma}
\begin{proof}
  Suppose $C$ contains two distinct points $p_i, p_i'$ on some line $L_i$. Since $C$ is minimally dependent, $\textup{span}\{p_i, p_i'\} = L_i$ intersects $\textup{span}(C \setminus \{p_i,p_i'\})$ in a single point $p_i'' \in L_i$. Replacing $p_i$ and $p_i'$ by $p_i''$, we obtain a circuit $C'$ such that:
  \begin{enumerate}
  \item $\textup{range}(C') = \textup{range}(C)$, 
  \item if $C$ contains only a single point $q_j$ on some line $L_j$, then $C'$ also contains $q_j$ (since $j \neq i$), and
  \item $|C'| = |C-1|$.
  \end{enumerate}
  We obtain the desired crossing circuit $D$ by repeating this process $|C| - \textup{range}(C)$ times: once for each line $L_i$ that contains two distinct points of $C$.
  
To prove the cardinality statement, assume $\textup{range}(D)=\{1,\dots, u\}$ For each $i=1,\dots,u$, let $P_i = [P_i^0: \dots : P_i^{k-1}]$ and $Q_i = [Q_i^0: \dots : Q_i^{k-1}]$ be the points of the ambient rational normal curve $Z$ that lie on $L_i$ and consider the $k \times 2u$ Vandermonde matrix
  \[ V = \begin{pmatrix} P_1^0 & Q_1^0 & \cdots & P_u^0 & Q_u^0 \\
    \vdots & \vdots & \ddots & \vdots & \vdots & \vdots \\
    P_1^{k-1} & Q_1^{k-1} & \cdots & P_u^{k-1} & Q_u^{k-1}    
  \end{pmatrix}. \]
  
  Any crossing circuit whose range is $L_1 \cup \dots \cup L_u$ consists of $u$ points
  $\alpha_1 P_1 + \beta_1Q_1, \dots, \alpha_u P_u + \beta_u Q_u \in \PP^{k-1}$ that support a unique linear dependence whose coefficients are all nonzero. 
Such a dependence defines a non-trivial element of the kernel of $V$. Since $V$ is injective whenever $2u\leq k$ the existence of a crossing circuit implies that $k<2u$ and thus $\lceil\frac{k+1}{2}\rceil\leq u$.
\end{proof}

\begin{proposition} \label{prop:bigcc}
  Suppose $Y$ is a dependent set of $k$ elements of $C(\FF_q)$ that does not contain any trivial circuit. Then there exists $u$ such that $\lceil\frac{k+1}{2}\rceil \leq u \leq \min\{k,m\}$ and a crossing circuit $D$ of size $u$  which contains at least $2u - k$ points of $Y$. 
  \end{proposition}
\begin{proof}
  Since $Y$ is dependent it contains a circuit $C\subseteq Y$, which is nontrivial by hypothesis. Let $c = |C|$ and $u = |\textup{range}(C)|$ and note that the inequalities $u\leq c\leq k$ and $u\leq m$ hold by construction. 
  
  Since $Y$ contains no trivial circuits, $C$ contains either one or two points on each of the lines in its range. By solving a $2 \times 2$ linear system we can see that there are $c - u$ lines containing exactly two points of $C$ and $2u - c$ lines containing exactly one point of $C$. 
    
  Now apply Lemma \ref{lem:existcc} to obtain a crossing circuit $D$ of size $u$ which shares at least $2u-c$ points with $C$. It follows that
  \[ |Y \cap D| \geq |C \cap D| = 2u - c \geq 2u - k. \]
  
  Finally, the inequality $\lceil\frac{k+1}{2}\rceil\leq u$ holds for the cardinality of every crossing circuit by the previous Lemma.
  \end{proof}

The contrapositive of Proposition \ref{prop:bigcc} yields the following effective criterion to verify that a set $\Gamma\subseteq C(\FF_q)$ leads to a PMDS code. It is dual to Lemma~\ref{lem:admissible_hyperplane} in that it does not describe the properties that admissible sets have but rather the structure of the subsets that these must avoid.

\begin{corollary}\label{cor:PMDScriterion}
Let $\Gamma$ be any set of elements of $C(\FF_q)$ such that for every $u$ satisfying $\lceil\frac{k+1}{2}\rceil \leq u \leq \min\{k,m\}$ and for every crossing circuit $D$ of size $u$, $|\Gamma \cap D|  \leq 2u-k-1$. Then every $k$-subset of $\Gamma$ that does not contain a trivial circuit is independent. In particular, if $\Gamma$ contains at least two points on each line $L_i$ then $\Gamma$ is admissible and the corresponding AG code $\calC$ is a PMDS code of dimension $k=2m-s$ (and localities $k_i=2$ for every $i$).
\end{corollary}
\begin{proof} Let $S\subseteq \Gamma'$ be an evaluation set of cardinality $k$. If $S$ were dependent then by Proposition~\ref{prop:bigcc} and our assumptions on $\Gamma$, $S$ would have to contain a trivial circuit, contradicting the fact that it is an evaluation set. We conclude that evaluation of linear forms at the points of $\Gamma'$ defines a PMDS code of dimension $k$ using Theorem~\ref{pmds_theorem}.
\end{proof}

\begin{example} \label{ex:4inP4}
  Let $m=4$ and $s=3$, so that $k=8-3=5$ and $C(\FF_q)$ consists of the $\FF_q$-points on 4 lines in $\PP^{k-1} = \PP^4$. Then $\frac{k+1}{2} = 3$, so we must consider crossing circuits of sizes 3 and 4. The case $u=3$ refers to collinear triples formed from one point on each of three different lines. We must choose $\Gamma$ so that for each such triple $D$, $|\Gamma \cap D| \leq 2\cdot 3 - 5 -1 = 0$. That is, we must exclude *all* points on the collinear triples. (Fortunately, we will show that there are very few such triples.)

  The case $u=4$ refers to coplanar quadruples formed from one point on each of the four lines. We must choose $\Gamma$ so that for each such quadruple $D$, $|\Gamma \cap D| \leq 2\cdot 4 - 5 -1 = 2$. 

  We can see directly that these restrictions are necessary by considering dependent 5-sets that would otherwise occur in $\Gamma$. Suppose $p_1 \in L_1, p_2 \in L_2, p_3 \in L_3$ form a collinear triple $D$. Then for any two points $q_1, r_1 \in L_1$ and any two points $q_2, r_2 \in L_2$, we have
  $p_3 \in \textup{span}\{p_1, p_2\} \subseteq \textup{span}(q_1,r_1,q_2,r_2)$. That is, $p_3$ lies in many dependent 5-sets, so must be excluded from $\Gamma$. By the same reasoning $p_1$ and $p_2$ must each be excluded from $\Gamma$; that is, $|\Gamma \cap D| = 0$.  
  
  Similarly, suppose that $D = \{p_1, p_2, p_3, p_4\}$ is a crossing circuit of size 4. Then for every pair of points $q_1, r_1 \in L_1$, $\{q_1,r_1,p_2,p_3,p_4\}$ is a dependent 5-set. (This set is in fact a circuit unless $p_1 \in \{q_1, r_1\}$, but either way it is dependent.) To avoid such sets, $p_2$, $p_3$, and $p_4$ must not all belong to $\Gamma$. That is, $|\Gamma \cap D| \leq 2$. 
\end{example}

In order to construct PMDS codes using Corollary \ref{cor:PMDScriterion}, we must bound the number of crossing circuits in $C(\FF_q)$ of each size.

\begin{proposition} \label{prop:countcc}
The number of crossing circuits of $C(\FF_q)$ whose range is a given set of lines of cardinality $u$ is $\begin{cases} 
    \mbox{at most } (q+1)^{2u-k-1} & \mbox{if } \frac{k+1}{2} \leq u \leq k \\
    0 & \mbox{otherwise.} \end{cases}$.
\end{proposition}
\begin{proof}
  Without loss of generality, assume that the range is $\{L_1, \dots, L_u\}$. For each $i$, let $P_i = [P_i^0: \dots : P_i^{k-1}]$ and $Q_i = [P_i^0: \dots : P_i^{k-1}]$ be the points of the rational normal curve $Z$ that lie on $L_i$ with fixed affine representatives. Consider the $k \times 2u$ Vandermonde matrix
  \[ V = \begin{pmatrix} P_1^0 & Q_1^0 & \cdots & P_u^0 & Q_u^0 \\
    \vdots & \vdots & \ddots & \vdots & \vdots & \vdots \\
    P_1^{k-1} & Q_1^{k-1} & \cdots & P_u^{k-1} & Q_u^{k-1}    
  \end{pmatrix}. \]

  Any crossing circuit whose range is $L_1 \cup \dots \cup L_u$ consists of $u$ points
  $\alpha_1 P_1 + \beta_1Q_1, \dots, \alpha_u P_u + \beta_u Q_u \in \PP^{k-1}$ that support a unique linear dependence $\lambda\in \left(\FF_q^*\right)^{u}$ whose coefficients are all nonzero. More precisely, fixing affine representatives for the points of our circuit we obtain unique coefficients $\alpha_1,\beta_1,\dots, \alpha_u,\beta_u$ expressing them as linear combinations of the $P_i$, $Q_i$. If $B(\alpha,\beta)$ denotes the $2u\times u$ matrix
  \[B(\alpha,\beta) = \begin{pmatrix} 
  \alpha_1 & 0 & \cdots & 0 & 0 \\
  \beta_1 & 0 & \cdots & 0 & 0 \\
 0 & \alpha_2 & 0 & \cdots & 0 \\
 0 & \beta_2 & 0 & \cdots & 0 \\
    \vdots & \vdots & \ddots & \vdots & \vdots & \vdots \\
    \vdots & \vdots & \ddots & \vdots & \vdots & \vdots \\
 0 & 0 & 0 & \cdots & \alpha_u \\
 0 & 0 & 0 & \cdots & \beta_u \\
  \end{pmatrix}. \]
  then the unique linear dependence $\lambda$ gives the equality $V B(\alpha,\beta) \lambda= 0$ and in particular the product $B(\alpha,\beta) \lambda $ is an element $w\in {\rm Ker}(V)\subseteq \FF_q^{2u}$. Since the product $B(\alpha,\beta)\lambda$ equals the transpose of the vector
$\left(\lambda_1\alpha_1, \lambda_1\beta_1,\dots, \lambda_u \alpha_u, \lambda_u \beta_u\right)$ we moreover conclude that knowing $[w]\in\PP^{2u}$ is enough to recover $[\alpha_1:\beta_1],\dots, [\alpha_u,\beta_u] \in \prod_{j=1}^u\PP^1$ because for nonzero $t$ the equality $[t\lambda_i\alpha_i: t\lambda_i\beta_i]=[\alpha_i: \beta_i]$ holds. In particular $[w]\in \PP\left({\rm Ker}(V)\right)$ allows us to recover at most one such crossing circuit. It follows that the total number of crossing circuits with the given range is bounded above by the number of points in $\PP({\rm ker}(V))\subseteq \PP(\FF_q^{2u})$. Since the Vandermonde matrix has rank $\min(2u,k)$ we conclude that this number is at most
  \[ \frac{q^{2u-k}-1}{q-1} = \sum_{j=0}^{2u-k-1}q^j \leq \sum_{j=0}^{2u-k-1} \binom{2u-k-1}{j}q^j = (q+1)^{2u-k-1}.\]
proving the claim.
\end{proof}

\begin{corollary} \label{cor:ccbound}
The number of crossing circuits of $C(\FF_q)$ of size $u$ is at most \[\binom{m}{u}(q+1)^{2u-k-1}.\]
\end{corollary}
\begin{proof}
This follows immediately from Proposition \ref{prop:countcc}.
\end{proof}

\begin{example} \label{ex:4onP4count}
  Continuing the case of 4 lines in $\PP^4$ as in Example \ref{ex:4inP4}, Corollary \ref{cor:ccbound} tells us that there are at most $\binom{4}{3} q^0 = 4$ crossing circuits of size three (collinear triples containing one point on each of three different lines) and at most $\binom{4}{4}q^2=q^2$ crossing circuits of size four (coplanar quadruples containing one point on each of the four lines). There are many other circuits formed by coplanar quadruples supported on just three lines, but by definition they are not crossing so they do not impose additional constraints in the setting of PMDS codes.  
\end{example}

\section{Random and Modified Random PMDS Codes Arising from Line Configurations} \label{sec:probabilistic}
In this section we will apply probabilistic methods and the results of the previous section in order to prove existence and abundance results for PMDS codes with certain parameters. In both cases our code is a subset of a reducible curve $C$ composed of the union of $m$ lines in $\PP^{k-1}$ as in Section \ref{sec:linecodes}. We fix the dimension $k$, the number of blocks $m$ and block localities all equal to two (and thus the global parameter $s = 2m-k$) and establish asymptotic bounds on the field size $q$ in terms of the length $n$. Throughout this section we assume that $k \geq m$, or equivalently that $s \leq m$.  

The first result, Theorem \ref{thm:randomcodes}, is based on taking a purely random subset of $C$ under a certain probability distribution, and yields a bound of $q = O(n^s)$. The argument is nearly constructive in the sense that for field sizes $q$ which satisfy our bounds we specify a probability distribution on $C(\FF_q)$ with the property that by selecting points independently according to this distribution we obtain a set $\Gamma$ with the remarkable property that it is admissible {\it with probability at least $1 - \epsilon$}. Sampling from this distribution is easy and thus allows us to construct PMDS codes in practice for all values of $s$. The second main result, Theorem \ref{thm:alterationcodes}, is obtained via the probabilistic method with alterations: we first pick a random subset under a different probability distribution and then remove certain points in order that what remains be an admissible set. This approach yields an improved lower bound of $q = O(n^{s-1})$ for the smallest field size over which such codes exist.

Although the second method gives a better field size bound, it is non-constructive, in the sense that we do not how to identify the points that must be removed without a brute-force search over exponentially many subsets of $C(\FF_q)$.

We will need a few basic results from probability theory which we now recall. The first result is the \emph{Markov inequality} for nonnegative random variables. This says that if $Y$ is any nonnegative random variable with finite expectation and $a > 1$, then
  \begin{equation} \label{eq:Markov} \Pr\left(Y \geq a E(Y)\right) \leq 1/a. \end{equation}
  
  In particular, if $Y$ is non-negative integer-valued, then by taking $a = \frac{1}{E(Y)}$ we obtain the following useful observation. 
  
\begin{proposition} \label{prop:integerRV}
If $Y$ is a non-negative integer-valued random variable, then $\Pr(Y = 0) > 1 - E(Y)$.
\end{proposition}

  The other result we need is a fairly sharp tail bound for the binomial distribution. Recall that a random variable $Y$ takes the binomial distribution $\textup{Binom}(N,p)$ if $\Pr(Y = k) = \binom{N}{k}p^k$. This happens if $Y$ is a sum of $N$ independent and identically distributed Bernoulli variables: each taking the value 1 with probability $p$ and 0 with probability $1-p$.
\begin{proposition} \label{prop:tailbounds} \cite{Janson}*{Theorem 1} If $Y \sim \textup{Binom}(N,p)$, then
\begin{equation}\label{eq:binomial_lower_bound} \Pr\left(Y \leq E(Y) - t\right) \leq \exp \left( \frac{-t^2}{2Np} \right). \end{equation}
\end{proposition}

\subsection{PMDS codes via random selection of evaluation points}
In this section our set $\Gamma$ will be a subset of the points of $C(\FF_q)$ obtained by randomly and independently selecting each point with a certain probability $p$. We will study the probability that the resulting set be admissible.

In view of Corollary \ref{cor:PMDScriterion}, the key random variables are:
\begin{itemize}
\item for each $i = 1, \dots, m$, the number $V_i$ of selected points  on the $i$-th line, and
\item for each $u$ such that $\frac{k+1}{2} \leq u \leq m$, the number $X_u$ of $(2u-k)$-subsets of crossing circuits of size $u$ such that the entire subset is selected. We also set $X = \sum_u X_u$.  
\end{itemize}

The goal is for $X_u$ to be zero for every $u$ so that Corollary \ref{cor:PMDScriterion} will apply, but we also would like the length $n = \sum_{i}  V_i$ to be as large as possible. These requirements are in tension with each other: increasing the probability $p$ will tend to increase both the $V_i$'s and the $X_u$'s. We thus begin by estimating the expected values of the random variables $X_u$ and $V_i$ in terms of $p$. In a second stage we will use concentration inequalities to control the deviation of our random variables from their mean. For notational clarity we let $Q:=q+1$.

\begin{proposition} \label{prop:VXexpectation}
Let $\Gamma$ be a subset of $C(\FF_q)$ obtained by selecting each point randomly and independently with probability $p = cQ^{- \alpha}$, where $c$ and $\alpha$ are any positive real numbers. Then
\begin{enumerate}
\item $E(V_i) = cQ^{1-\alpha}$,
\item $E(X_u) \leq  c^j \binom{m}{j,u-j,m-u} Q^{(1- \alpha)j - 1}$, where $j = 2u-k = 2u-2m+s$.
\end{enumerate}
\end{proposition}
\begin{proof} $ $
  \begin{enumerate} 
  \item The random variable $V_i$ takes the distribution  $\textup{Binom}(Q,p)$ so its expected value is $pQ = (cQ^{ - \alpha})Q = cQ^{1-\alpha}$.
   \item For a given crossing circuit $C$ of size $u$ and a given $j$-subset $J \subseteq C$, let $Z_J$ be the indicator variable for the event that $J \subseteq \Gamma'$. Then
  \[ E(Z_j) = \Pr(Z_j = 1) = p^j = c^jQ^{- \alpha j} .\]
  Now by Corollary \ref{cor:ccbound} there are at most
  \[ \binom{m}{u} Q^{2u-k-1} = \binom{m}{u} Q^{j-1}\]
  crossing circuits of size $u$ to be considered, and each of them contains  $\binom{u}{j}$ $j$-subsets. Thus by linearity of expectation,
  \begin{align*}
 E(X_u) & \leq c^j Q^{- \alpha j}\binom{m}{u} Q^{j-1} \binom{u}{j}  \\
    = c^j \binom{m}{u} \binom{u}{j} Q^{-\alpha j + j - 1} 
    & = c^j \binom{m}{j,u-j,m-u} Q^{(1 - \alpha)j - 1}. 
  \end{align*}
  \end{enumerate}
\end{proof}

From Proposition \ref{prop:VXexpectation}, we see that $E(X_m)$ asymptotically dominates $E(X_u)$ for all smaller values of $u$. So our strategy will be to choose $\alpha$ and $c$ as large as possible under the constraint that $X_m$ must be zero with substantial probability. 

\begin{lemma} \label{lem:Xbounded} Let $0 < \varepsilon < 1$ and $c = \left( \frac{\varepsilon}{3 \binom{m}{s}} \right)^{\frac{1}{s}}$, $\alpha = 1 - \frac{1}{s}$ and $p = cQ^{- \alpha}$. Then \\
  $E(X) \leq \frac{\varepsilon}{3} + 3^m Q^{\frac{-2}{s}}$
\end{lemma}
\begin{proof} We first consider $X_m$. For $u=m$ we have $j=2u-2u+s = s$, and so Proposition \ref{prop:VXexpectation} (b) gives
  \begin{align*} E(X_m) & \leq \left( \left( \frac{\varepsilon}{3 \binom{m}{s}} \right)^\frac{1}{s} \right)^s \binom{m}{s,m-s-0} Q^{1 - \left(1 - \frac{1}{s}\right)s - 1} \\
      = \frac{\varepsilon}{3 \binom{m}{s}} \binom{m}{s} Q^0 
        & = \frac{\varepsilon}{3}.
  \end{align*}     
  
Now let $X^{-} = X - X_m = \sum_{u < m} X_u$. By Proposition \ref{prop:VXexpectation} (b) and linearity of expectation,
\[ E(X^-) \leq \sum_{\frac{k+1}{2} \leq u \leq m-1} c^j \binom{m}{j,u-j,m-u} Q^{1 - \left(1 - \frac{1}{s}\right)j - 1}.  \]
  
  Now $c < 1$ and $j$ is always positive, so $c^j < 1$ in every term. Among all the terms, the highest power of $Q$ is obtained when $u = m-1$, in which case $j = 2(m-1) - 2m + s = s-2$. The resulting exponent is $\frac{1}{s}\left( s-2 \right) - 1 = \frac{-2}{s}$. Thus
\begin{align*}
  E(X^-) & \leq  Q^{\frac{-2}{s}} \sum_{\frac{k+1}{2} \leq u \leq m-1} \binom{m}{j,u-j,m-u} \\
  &\leq Q^{\frac{-2}{s}} \sum_{a+b+c = m} \binom{m}{a,b,c} 
   = 3^m Q^{\frac{-2}{s}}.
\end{align*}

Then $E(X) = E(X_m) + E(X^-) \leq \frac{\varepsilon}{3} + 3^m Q^{\frac{-2}{s}}$. 
\end{proof}

Now that we have fixed $p$, Proposition \ref{prop:VXexpectation} (a) tells us the expected number of selected points on each line. We next prove that with substantial probability, all of the lines actually have almost this many selected points.
\begin{lemma} \label{lem:Vsmall}
  Let $c$, $\alpha$, and $p$ be as in Lemma \ref{lem:Xbounded} and let
  \[t = \sqrt{-2c \ln \left( 1 - \left( 1 - \frac{\varepsilon}{3} \right)^{\frac{1}{m}} \right)} Q^{\frac{1}{2s}}.\]
  Then for each $i$,
  \[ \Pr \left( V_i \leq cQ^{\frac{1}{s}} - t \right) \leq 1 - \left( 1 - \frac{\varepsilon}{3} \right)^{\frac{1}{m}} .\]
\end{lemma}
\begin{proof}
  By Proposition \ref{prop:VXexpectation}, the random variable $V_i$ is binomial with expectation $E(V_i) = cQ^{1 - (1 - \frac{1}{s})} = cQ^{\frac{1}{s}}$. So by applying the tail bound (\ref{eq:binomial_lower_bound}), we obtain
  \begin{align*}
    \Pr(V_i < cQ^{\frac{1}{s}} - t) &
    \leq \exp \left( \frac{2c \ln \left( 1 - \left( 1 - \frac{\varepsilon}{3} \right)^{\frac{1}{m}} \right) Q^{\frac{1}{s}}}{2cQ^\frac{1}{s}} \right) \\
& = \exp \left( \ln \left( 1 - \left( 1 - \frac{\varepsilon}{3} \right)^{\frac{1}{m}} \right) \right) 
 = 1 - \left( 1 - \frac{\varepsilon}{3} \right)^{\frac{1}{m}}.
  \end{align*}
\end{proof}

We can now prove the first main theorem of this section. 
\begin{theorem} \label{thm:randomcodes}
  Let $C$ be a reducible curve composed of $m$ generic lines in $\PP^{k-1}$, where $\frac{k+1}{2} \leq m \leq k$. Set $s=2m-k$ and $c= \left( \frac{\varepsilon}{3 \binom{m}{s}} \right)^{\frac{1}{s}}$. Let $\Gamma$ be a subset of $C(\FF_q)$ obtained by selecting each point randomly and independently with probability $p = cQ^{1 - \frac{1}{s}}$, where $Q = q+1$. If 
 \[ \label{eq:minQ} Q \geq  \left( \frac{3^{m+1}}{\varepsilon} \right)^{\frac{s}{2}}   \]
  and if $n$ is a positive integer such that 
  \[ \label{eq:boundQ} cQ^\frac{1}{s} - \sqrt{-2c \ln \left( 1 - \left( 1 - \frac{\varepsilon}{3} \right)^{\frac{1}{m}} \right)} Q^{\frac{1}{2s}} \geq \frac{n}{m}, \]
  then with probability at least $1 - \varepsilon$, $\Gamma$ is a PDMS code of length at least $n$, dimension $k$, $m$ blocks with block locality $(2, \dots 2)$ and global parameter $s$. 
\end{theorem}
\begin{proof}
  In light of Corollary \ref{cor:PMDScriterion}, there are three conditions that must be satisfied: $\Gamma$ must contain at least two points on each line, the length must be at least $n$, and for each $u$ and for each crossing circuit $D$ of size $u$, we must have $|\Gamma' \cap D|  \leq 2u-k-1$. For the first two conditions both to hold, it will be sufficient to have $V_i \geq \frac{n}{m}$ for $i=1,\dots,m$. The third condition is simply that $X=0$.
  
  We begin with the condition on $X$. The first hypothesis is equivalent to $3^m q^{\frac{-2}{s}} \leq \frac{\varepsilon}{3}$, and since $Q = q+1$, we obtain $3^m Q^{\frac{-2}{s}} < \frac{\varepsilon}{3}$. Then by Proposition \ref{prop:integerRV} and Lemma \ref{lem:Xbounded}, we have
 \[  \Pr(X > 0) \leq E(X) \leq  \frac{\varepsilon}{3} + 3^m Q^{\frac{-2}{s}} < \frac{2 \varepsilon}{3} .\] 

 For $V_i$, the second hypothesis can be written as $cQ^\frac{1}{s} - t \geq \frac{n}{m}$, where $t$ is as in Lemma \ref{lem:Vsmall}. So by this lemma, we have for each $V_i$ that $ \Pr(V_i \leq \frac{n}{m}) \leq 1 - \left( 1 - \frac{\varepsilon}{3} \right)^{\frac{1}{m}}$, or equivalently $\Pr(V_i > \frac{n}{m}) \geq  \left( 1 - \frac{\varepsilon}{3} \right)^{\frac{1}{m}}$. Now the events $V_1, \dots V_m$ are mutually independent since they involve selections of points on different lines, so the probability that \emph{every} $V_i$ is at least $\frac{n}{m}$ is at least $1 - \frac{\varepsilon}{3}$. (That is, the probability that some $V_i$ is too small is at most $\frac{\varepsilon}{3}$.)
 
By the union bound, the probability of failure is at most $\frac{2 \varepsilon}{3} + \frac{\varepsilon}{3} = \varepsilon$ and so we obtain a code with all of the desired properties with probability at least $1 - \varepsilon$.
\end{proof}

\begin{corollary} \label{cor:randomcodes} For fixed $k \geq m$ and growing $n$, and taking $s=2m-k$, there exist PDMS codes of length at least $n$, dimension $k$, $m$ blocks with block locality $(2, \dots 2)$ and global parameter $s$ over $\FF_q$ for $q = O(n^s)$.   
\end{corollary}
\begin{proof}
The inequality (\ref{eq:minQ}) in Theorem \ref{thm:randomcodes} is satisfied for all sufficiently large $q$, and the inequality (\ref{eq:boundQ}) requires $n$ to be at least as large as a function of $q$ whose leading term is a multiple of $q^{\frac{1}{s}}$. All of this can be satisfied if $q = O(n^s)$. 
\end{proof}

\subsection{Improved asymptotics via alterations}
We will now improve the asymptotic bound on the field size from $q = O(n^s)$ to $q = O(n^{s-1})$ via the \emph{probabilistic method with alterations} (see \cite{AS}*{Chapter 3]}) which works as follows.

Given a finite set $\Omega$, we aim to show that there exists a a subset $\Delta \in \Omega$ of a certain size that contains no ``bad'' substructures of a given type. To do this, we form a subset of $\Omega$ by selecting elements randomly and independently with a certain probability, and let $Y$ be a random variable that represents the number of bad substructures in the selected subset. If we can show that $Y=0$ with nonzero probability, then we are done. This is the pure probabilistic method that we used in the previous subsection.

But sometimes it is not possible to show directly that $Y=0$ with nonzero probability. Instead, we can show that $Y$ is reasonably small with nonzero probability. Then we remove enough elements of $\Delta$ to form a subset $\Delta' \subseteq \Delta$ with no bad substructures that is still reasonably large.

\begin{remark} A beautiful application of this method \cite{AS}*{pp. 38--39}
is the proof of Erd\H{o}s that there exist graphs $G$ such that the chromatic number of $G$ and the length of the shortest cycle in $G$ are both arbitrarily large.
\end{remark}

As before, we let $C$ be a reducible curve composed of $m$ lines in $\PP^{k-1}$, where $\frac{k+1}{2} \leq m \leq k$. Our first step is again to form a subset $\Gamma$ of $C(\FF_q)$ by randomly and independently selecting each point with a certain probability $p$. But this time we will take a slightly larger value of $p$ in order to increase the number of selected points. The cost of increasing $p$ is that $E(X)$ is no longer bounded by a constant, so we cannot hope that $X = 0$ which would guarantee $\Gamma$ itself is admissible. Instead, we will follow the strategy outlined above: that is, we remove $|X|$ points to obtain an admissible set $\Gamma'$ and show that with substantial probability, $\Gamma'$ still contains many points on each line. 

Specifically, in this section we will take
\[ p = cQ^{-\frac{s-2}{s-1}} \]
where $c < 1$ is a constant to be specified later. We begin by computing the expected values of the key random variables $V_i$, $X_u$ for each $u$, and $X = \sum_u X_u$ for this choice of $p$.

\begin{lemma} \label{lem:bestp}
  Let $p=cQ^{- \frac{s-2}{s-1}}$ where $0 < c < 1$.
\begin{enumerate}
\item For each $i$, we have $E(V_i) \geq cQ^\frac{1}{s-1}.$
  \item For each $u$, if again $j = 2u-2m+s$ then     
    \[ E(X_u) \leq c^j \binom{m}{j,u-j,m-u} Q^{\frac{j-s+1}{s-1}}.\]
  \item In particular, $E(X_m) \leq c^s \binom{m}{s}Q^{\frac{1}{s-1}}$.
  \item $E(X) \leq c^s \binom{m}{s}Q^{\frac{1}{s-1}} + 3^m Q^{- \frac{1}{s-1}}$.   
\end{enumerate}
\end{lemma}
\begin{proof} $ $
  \begin{enumerate}
  \item We apply Proposition \ref{prop:VXexpectation} (a) with $\alpha = \frac{s-2}{s-1}$ to obtain
    \[E(V_i) = cQ^{1 - \frac{s-2}{s-1}} = cQ^{\frac{1}{s-1}}.\]
    \item We apply Proposition \ref{prop:VXexpectation} (b) with $\alpha = \frac{s-2}{s-1}$ to obtain
      $$E(X_u) =  c^j \binom{m}{j,u-j,m-u} Q^{\left( 1 - \frac{s-2}{s-1} \right)j - 1} = c^j \binom{m}{j,u-j,m-u} Q^{\frac{j-s+1}{s-1}}.$$
    \item For $u = m$, we have $j = 2m-2m+s = s$, so this follows from the preceding statement.
    \item As before, let $X^- = \sum_{u < m} X_u$. Then by (b) and linearity of expectation we have
      \[ E(X^-) \leq \sum_{\frac{k+1}{2} \leq u \leq m-1}c^j \binom{m}{j,u-j,m-u} Q^{\frac{j-s+1}{s-1}}.\]
      Since $c < 1$, we have $c^j < 1$ for each $j$. The highest power of $Q$ is obtained when $u = m-1$, in which case $j=2(m-1)-2m+s = s-2$. The resulting exponent of $Q$ is $\frac{s-2-s+1}{s-1} = \frac{-1}{s-1}$. Thus
\begin{align*} 
  E(X^-) & \leq Q^{-\frac{1}{s-1}} \sum_{\frac{k+1}{2} \leq u \leq m-1} \binom{m}{j,u-j,m-u} \\
  &\leq Q^{-\frac{1}{s-1}} \sum_{a+b+c = m} \binom{m}{a,b,c} 
   = 3^m Q^{\frac{-2}{s}}.
\end{align*}
Finally, $E(X) = E(X_m) + E(X^-)  \leq c^s \binom{m}{s}Q^{\frac{1}{s-1}} + 3^m Q^{- \frac{1}{s-1}}$.
  \end{enumerate}
\end{proof}

\begin{remark}
If we were to choose $p = cQ^{- \beta}$ with $\beta > \frac{s-2}{s-1}$ then $E(V_i) = o\left(Q^\frac{1}{s-1}\right)$. On the other hand, if $p = cQ^{- \beta}$ with $\beta < \frac{s-2}{s-1}$, then $E(X_m)$ and hence $E(X)$ would grow faster than $E(V_i)$. That is, there would be more bad substructures than points and the probabilistic method with alterations would not work. Thus the exponent $\alpha = \frac{s-2}{s-1}$ is optimal for this method. 
\end{remark}

Recall that for each $u$, $X_u$ counts the number of ($2u-k$)-subsets of crossing circuits of size $u$ that are selected in $\Gamma$. We must remove one point from $\Gamma$ for each such subset (including ranging over all the different values of $u$.) Unfortunately, different subsets intersect different collections of lines and it is not clear how to remove points evenly from all from the lines. Instead, we simply work from the worst-case assumption that we always remove points from the same line $L_i$. In this case, the collection $\Gamma'$ of remaining selected points contains $V_i - X$ points on $L_i$. So our objective is to show that with substantial probability, $V_i - X$ is still reasonably large for every $i$.  
\begin{proposition} \label{prop:fullbound}
  Let $ p = cQ^{-\frac{s-2}{s-1}}$ for any $0 < c <1$. With probability at least $\frac{1}{6}$, we have for every $i=1,\dots,m$ that
  \[ V_i - X > \left( c - 2c^s\binom{m}{s} \right) Q^{\frac{1}{s-1}} - \sqrt{2c\log(3m)}Q^\frac{1}{2(s-1)} - 2 \cdot 3^m Q^\frac{-1}{s-1}.\]
\end{proposition}
\begin{proof}
For each binomial random variable $V_i$, we apply the tail bound (\ref{eq:binomial_lower_bound}). Choosing $t = bQ^\frac{1}{2(s-1)}$ where $b=\sqrt{2c\log(3m)}$ yields 
\begin{align*}
 \Pr\left(V_i \leq cQ^\frac{1}{s-1} - bQ^\frac{1}{2(s-1)}\right) & \leq \exp{\left(\frac{-b^2Q^\frac{1}{s-1}}{2cQ^\frac{1}{s-1}}\right)} \\
 &= \exp{\left( \frac{-b^2}{2c}\right)} 
 = \exp{\left( -\log{(3m)} \right)} 
 = \frac{1}{3m}.
\end{align*}
Then, by the union bound, we obtain that 
\begin{equation} \label{eq:Vbound_general} 
\Pr\left(\exists i: \, V_i \leq cQ^\frac{1}{s-1} - bQ^\frac{1}{2(s-1)}\right) \leq \frac{1}{3}.
\end{equation}

For $X$, we use Markov's inequality (\ref{eq:Markov}) with $a=2$ to obtain
\begin{equation} \label{eq:Xbound_general}\Pr \left(X \geq 2E(X) \right) \leq 1/2. \end{equation}

Combining (\ref{eq:Xbound_general}) and (\ref{eq:Vbound_general}) and considering the complements of the two events, we conclude that with probability at least $1-\frac{1}{2}-\frac{1}{3} = \frac{1}{6}$, we have for every $i$ that
\begin{align*}
V_i - X & \geq cQ^\frac{1}{s-1} - bQ^\frac{1}{2(s-1)} - 2E(X) \\ 
& = cQ^\frac{1}{s-1} - bQ^\frac{1}{2(s-1)} - 2E(X_m) - 2E(X^-) \\
& \geq cQ^\frac{1}{s-1} - bQ^\frac{1}{2(s-1)} - 2 \cdot c^s \binom{m}{s}Q^{\frac{1}{s-1}}  - 2 \cdot 3^m Q^\frac{-1}{s-1} \\
& = \left( c - 2\binom{m}{s}c^s \right) Q^{\frac{1}{s-1}} - \sqrt{2c\log(3m)}Q^\frac{1}{2(s-1)} - 2 \cdot 3^m Q^\frac{-1}{s-1}.
\end{align*}
\end{proof}

The next step is to choose a value of $c$. Again we will consider only the leading power of $Q$ in attempting to optimize our bound. 

\begin{lemma} \label{lem:bestc}
  Let $f(c) = c - ac^s$ with $s > 1$ and $a > 0$. This function takes its unique maximum at $c = \left(as\right)^\frac{1}{1-s}$ and its maximum value is $\left(as\right)^{\frac{1}{1-s}} \left( 1 - s^{-1} \right)$.  
\end{lemma}
\begin{proof}
  This follows from elementary differential calculus, with $f'(c)=1-sac^{s-1}$ and $f'(c)=0 \iff  c=(as)^{-\frac{1}{s-1}}$. 
\end{proof}

\begin{theorem} \label{thm:alterationcodes}
  Let $C$ be a reducible curve composed of $m$ generic lines in $\PP^{k-1}$, where $\frac{k+1}{2} \leq m \leq k$. Set $s=2m-k$ and $c = \left(as\right)^\frac{1}{1-s}$. Let $\Gamma$ be a subset of $C(\FF_q)$ obtained by selecting each point randomly and independently with probability $p = cQ^{-\frac{s-2}{s-1}}$. If
  \[ \left( c - 2\binom{m}{s}c^s \right) Q^{\frac{1}{s-1}} - \sqrt{2c\log(3m)}Q^\frac{1}{2(s-1)} - 2 \cdot 3^m Q^\frac{-1}{s-1} \geq \frac{n}{m} \] 
  then with probability at least $\frac{1}{6}$ there exists $\Gamma' \subset \Gamma$ such that $\Gamma'$ is a PDMS code of length at least $n$, dimension $k$, $m$ blocks with block locality $(2, \dots, 2)$ and global parameter $s$. 
\end{theorem}
\begin{proof}
This follows immediately from Proposition \ref{prop:fullbound}: we can remove $|X|$ points from \emph{every} line and obtain $\Gamma' \subset \Gamma$ that still contains at least $\frac{n}{m}$ points on each line and $\Gamma'$ is a PMDS code with the desired parameters for the same reason as in Theorem \ref{thm:randomcodes}.
\end{proof}

\begin{corollary} \label{cor:alterationcodes} For fixed $m \leq k$ and growing $n$, and taking $s=2m-k$, there exist PDMS codes of length at least $n$, dimension $k$, $m$ blocks with block locality $(2, \dots,  2)$ and global parameter $s$ over $\FF_q$ for $q = O(n^{s-1})$.   
\end{corollary}
\begin{proof}
Just as in Corollary \ref{cor:randomcodes}, this follows from asymptotically inverting the bound in Theorem \ref{thm:alterationcodes}, but this time $n$ only needs to be greater than a function growing like a multiple of $q^{\frac{1}{s-1}}$ rather than $q^{\frac{1}{s}}$.  
\end{proof}

We conclude the section with a discussion about preexisting results on the field size required for construction PMDS codes.

\begin{remark}
\label{rmk: comparison}
 Let us compare the field size obtained in Corollary \ref{cor:alterationcodes} for the existence of algebraic-geometric PMDS codes with  known existence results in literature. For simplicity, we do it restricting to the case of homogeneous PMDS codes, that is when $n_1=\ldots=n_m=\frac{n}{m}$ and when $k_1=\ldots=k_m=2$. Observe that our result requires that $s \leq m$. We concentrate on the asymptotic case for $n$, while we consider $m$ and $s$ fixed.
 
 One of the first general results on the field size required for PMDS codes was given by Chen \textit{et al.} in \cite{ch07}, where the author proved that for $q=O(\binom{n-1}{2m-s-1})$ we can always have PMDS codes. For fixed $s$ and $m$ with $s\leq m$,  $\binom{n-1}{2m-s-1}$ is a polynomial in $n$ of degree $2m-s-1\geq s-1$, hence Corollary \ref{cor:alterationcodes} is an improvement on this result for the asymptotic regime.
 
 Now, we consider the first general construction of PMDS codes for every value $s$ due to Calis and Koyluoglu in \cite{ca17}, based on Gabidulin codes. The same construction was also provided in \cite{rawat2013optimal}, although the authors did not investigate the PMDS structure. The field size required is $q=O((\frac{n}{m})^{2m})$. Since $s \leq m$, this field size  in the asymptotic regime is worse than the one of Corollary \ref{cor:alterationcodes}.
 
 Another general construction was presented by Gabrys \textit{et al.} in \cite{gabrys2018constructions}. For that construction, the size needed for the underlying field is $q=O(\max \{ m, (\frac{n}{m})^{\frac{n}{m}}\}^s)$,  which is exponential in $n$ for fixed $m$, and hence, the field size given in Corollary  \ref{cor:alterationcodes} is much better. 
 
 A last construction that we need to mention is the one provided by Mart{\'\i}nez-Pe{\~n}as and Kschischang in \cite{martinez2019universal}. It is based on linearized Reed-Solomon codes in the sum-rank metric, which are a generalization of Gabidulin codes. This allows to sensibly reduce the field size, which is shown to be $q=O(\max\{m+1,\frac{n}{m}\}^2)$. In the asymptotic regime, Corollary \ref{cor:alterationcodes} is better only when $s=2$, for which we also give an explicit construction in Theorem \ref{thm:PMDSconstruction_s=2}, and it is comparable when $s=3$. However, their construction does not work over prime fields, while our approach does not present this obstruction.
\end{remark}

 \bibliographystyle{abbrv}
\bibliography{PMDS.bib}
\end{document}